\documentclass[11pt]{article}

\usepackage[T1]{fontenc}
\usepackage[margin=1in]{geometry}
\usepackage{mathtools,amssymb,amsthm}
\usepackage{microtype}
\usepackage{graphicx,booktabs,enumitem}
\usepackage{xcolor}
\usepackage{tikz}
\usetikzlibrary{arrows.meta}
\usepackage[hidelinks]{hyperref}
\usepackage[capitalize,nameinlink,noabbrev]{cleveref}

\theoremstyle{plain}
\newtheorem{theorem}{Theorem}[section]
\newtheorem{lemma}[theorem]{Lemma}
\newtheorem{corollary}[theorem]{Corollary}

\newtheorem{claim}[theorem]{Claim}
\newtheorem{fact}[theorem]{Fact}

\theoremstyle{definition}
\newtheorem{definition}[theorem]{Definition}

\theoremstyle{remark}

\numberwithin{equation}{section}

\newcommand{\F}{\mathbb{F}}
\newcommand{\E}{\mathbb{E}}
\newcommand{\bits}{\{0,1\}}
\newcommand{\ind}{\mathbf{1}}

\newcommand{\NP}{\mathsf{NP}}
\newcommand{\BPP}{\mathsf{BPP}}
\newcommand{\RP}{\mathsf{RP}}
\newcommand{\PRP}{\mathsf{P}^{\RP}}
\newcommand{\PNP}{\mathsf{P}^{\NP}}

\DeclareMathOperator{\poly}{poly}
\DeclareMathOperator{\polylog}{polylog}
\newcommand{\GH}{\mathsf{GapHam}}
\newcommand{\indicate}{\mathsf{Ind}}

\newcommand{\cfunc}{\mathsf{Check}}
\newcommand{\adrs}{\mathsf{Ads}}
\newcommand{\andgadget}{\mathsf{AND}}
\newcommand{\tGH}{\widetilde{\mathsf{GapHam}}}

\newcommand{\PNPcc}{\mathsf{P}^{\NP cc}}
\newcommand{\Rcc}{\mathsf{R}}
\newcommand{\rect}{\mathsf{rect}}
\newcommand{\PRPcc}{\mathsf{P}^{\RP cc}}

\newcommand{\ba}{\mathbf{a}}
\newcommand{\bb}{\mathbf{b}}
\newcommand{\bc}{\mathbf{c}}

\newcommand{\bu}{\mathbf{u}}
\newcommand{\bv}{\mathbf{v}}

\newcommand{\bx}{\mathbf{x}}
\newcommand{\by}{\mathbf{y}}
\newcommand{\bz}{\mathbf{z}}

\newcommand{\bZ}{\mathbf{Z}}

\title{Efficient Randomized Communication Without Large Monochromatic Rectangles}
\author{%
  Haoyu Wang\\
  \small Penn State University\\
  \small\href{mailto:hjw5492@psu.edu}{\texttt{hjw5492@psu.edu}}
  \and
  Pei Wu\\
  \small Penn State University\\
  \small\href{mailto:pei.wu@psu.edu}{\texttt{pei.wu@psu.edu}}%
}
\date{}

\begin{document}
\maketitle

\begin{abstract}
In this paper, we construct a total Boolean function with  $\widetilde{O}(\log n)$ randomized communication protocol, while any monochromatic rectangle has density at most  $O(2^{-\mathrm{poly}(n)})$. As a corollary, it gives the first total function separation for $\BPP\not\subseteq\PNP$ in the communication world. Inspired by Gavinsky’s recent work (arXiv:2608.18784), our construction combines the cheat-sheet framework with fully linear PCPs.
\end{abstract}

\section{Introduction}
\label{sec:intro}

In the two-party communication model of Yao~\cite{yao1979complexity},
Alice and Bob receive inputs
$x,y\in\bits^n$, respectively, and jointly compute a Boolean function
$F(x,y)$. The goal is to communicate as few bits as possible. The randomized class
$\BPP$ in communication complexity consists of families of functions computable using public
randomness and at most $\polylog n$ bits of communication, with error
at most $1/3$ on every input.

Understanding randomness relative to nondeterminism and the
polynomial hierarchy is a central question in complexity theory.
In communication complexity, $\BPP$ is incomparable with
both $\NP$ and $\mathsf{coNP}$~\cite{babai1986complexity}, whereas the corresponding
computational comparisons remain open. At the second level of the polynomial hierarchy, the
Sipser--G\'acs--Lautemann argument carries over to communication, as
recorded by Babai, Frankl, and
Simon~\cite[Proposition~7.8]{babai1986complexity}:
\[
 \BPP\subseteq\Sigma_2\cap\Pi_2.
\]
This containment is strict, since $\NP\subseteq\Sigma_2\cap\Pi_2$ and
$\NP\not\subseteq\BPP$.
This picture motivates a finer comparison within $\Sigma_2\cap\Pi_2$, between $\BPP$ and
$\PNP$, the class of efficient deterministic protocols with adaptive
access to $\NP$ predicates. These $\NP$ queries return exact answers,
which can determine the choice of later queries. The question of
whether $\BPP\subseteq\PNP$ holds for total functions was explicitly
raised by G\"o\"os, Pitassi, and
Watson~\cite{goos2018landscape}. Hatami and
Hatami~\cite{hatami2026structure} conjecture that this
containment fails.

A second motivation is to relate randomized communication cost to
exact combinatorial structure in Boolean matrices. Let $\Rcc(F)$
denote the minimum cost of a public-coin randomized protocol that
computes $F$ with error at most $1/3$ on every input. View $F$ as a
matrix whose rows and columns are indexed by Alice's and Bob's inputs.
A monochromatic rectangle is a product set $A\times B$, with
$A,B\subseteq\bits^n$, on which $F$ is constant. Write
\[
 \rect(F)
 =
 \max\left\{
 \frac{|A||B|}{2^{2n}}:
 F\text{ is constant on }A\times B
 \right\}.
\]
Does every Boolean function with an efficient randomized protocol
contain a large monochromatic rectangle?
The log-rank conjecture provides a familiar example of such a
structural question. Nisan and Wigderson~\cite{nisan1995rank} showed
that this conjecture is equivalent to the universal assertion
$\rect(F)\ge 2^{-(\polylog(\operatorname{rank}(F)+2))}$, where $\operatorname{rank}(F)$ is the
 real rank of the communication matrix. Following their
separation of $\BPP$ from deterministic protocols with adaptive
Equality queries, Chattopadhyay, Lovett, and
Vinyals~\cite[Problem~6.1]{chattopadhyay2019equality} asked the
analogous question for randomized protocols: does
$\rect(F)\ge 2^{-O(\Rcc(F))}$ always hold?
Hatami and Hatami~\cite{hatami2026structure} asked this question again in their recent survey.

A structural property of $\PNP$ connects these questions.
Papakonstantinou, Scheder, and
Song~\cite{papakonstantinou2014overlays} characterized $\PNP$ through
\emph{rectangle decision lists}, which they called \emph{rectangle
overlays}. A list consists of labeled rectangles covering the input
domain, and each input receives the label of the first rectangle
containing it. The logarithm of the minimum list length is polynomially
related to the cost of an optimal $\PNP$ protocol. The \emph{product
method} of Impagliazzo and Williams~\cite{impagliazzo2010communication}
provides a lower-bound criterion for this model: every protocol of
cost $c$ has a monochromatic rectangle of measure at least $2^{-O(c)}$
under every product distribution on the two inputs. Taking the
uniform distribution shows that every function in $\PNP$ satisfies
$\rect(F)\ge 2^{-\polylog n}$.

We construct a \emph{total} Boolean function with small randomized
communication complexity and no large monochromatic rectangle,
resolving both questions.

\subsection{Our results}

\begin{theorem}\label{thm:main}
There is a family of total Boolean functions
$F_n:\bits^n\times\bits^n\to\bits$ satisfying
\[
 \Rcc(F_n)=O(\log n\log\log n)
 \qquad\text{and}\qquad
 \rect(F_n)
 \le 2^{-n^{\Omega(1)}}.
\]
\end{theorem}

The family is obtained from the cheat-sheet function framework on the Gap-Hamming problem
$\tGH$ constructed in \cref{sec:construction}.

The totality requirement is essential to the class separation.
Gap-Hamming with a gap linear in its input length already separates randomized
communication from adaptive $\NP$ queries on a promise domain:
random sampling gives a constant-cost protocol, whereas the lower
bounds of Papakonstantinou, Scheder, and
Song~\cite{papakonstantinou2014overlays} apply to every total extension.
The difficulty is to define the function outside its promise while
retaining an efficient randomized protocol.

Write $\PNPcc(F)$ for the minimum cost of a deterministic protocol
with adaptive $\NP$ queries, charged their nondeterministic
communication cost. Precise cost conventions appear in
\cref{sec:preliminaries}.
The product method of Impagliazzo and
Williams~\cite{impagliazzo2010communication} gives
\[
 \PNPcc(F)
 =
 \Omega\!\left(\log\frac{1}{\rect(F)}\right).
\]
Applying this bound under the uniform distribution to
\cref{thm:main} yields the following separation.

\begin{corollary}\label{cor:bpp-pnp}
The family $\{F_n\}$ in \cref{thm:main} satisfies
\[
 \Rcc(F_n)=O(\log n\log\log n)
 \qquad\text{and}\qquad
 \PNPcc(F_n)\ge n^{\Omega(1)}.
\]
In particular, $\BPP\not\subseteq\PNP$ for total functions.
\end{corollary}
Since $\NP\not\subseteq\BPP$ was already known
\cite{babai1986complexity}, Corollary~\ref{cor:bpp-pnp} establishes that $\BPP$
and $\PNP$ are incomparable within $\Sigma_2\cap\Pi_2$, even for
total functions.

\medskip
The second corollary concerns $\RP$, the one-sided-error version of $\BPP$: protocols
never err on $0$-inputs and have error at most $1/3$ on $1$-inputs.
We have the standard containment
\[
 \PRP\subseteq\BPP\cap\PNP.
\]
Indeed, error amplification allows efficient adaptive $\RP$-query
protocols to be simulated in $\BPP$, while $\RP\subseteq\NP$ gives
the containment in $\PNP$~\cite{hatami2026structure}.
Write $\PRPcc(F)$ for the analogous cost with exact $\RP$ queries,
charged their one-sided randomized communication cost.
Since an $\RP$ query of cost $h\ge1$ has nondeterministic cost at most
$h+O(\log n)$~\cite{hatami2026structure}, simulating each query gives
$\PNPcc(F)\le O(\log n)\,\PRPcc(F)$.
Corollary~\ref{cor:bpp-pnp} then gives the following
quantitative consequence.

\begin{corollary}\label{cor:bpp-prp}
The family $\{F_n\}$ in \cref{thm:main} satisfies
\[
 \Rcc(F_n)=O(\log n\log\log n)
 \qquad\text{and}\qquad
 \PRPcc(F_n)\ge n^{\Omega(1)}.
\]
In particular, $\PRP\subsetneq\BPP$ for total functions.
\end{corollary}

Corollary~\ref{cor:bpp-prp} shows that even adaptive composition of
efficient one-sided-error randomized protocols cannot simulate every
efficient two-sided-error randomized protocol, answering a question
of Pitassi, Shirley, and Watson~\cite{pitassi2020nondeterministic}.
Figure~\ref{fig:class-relations} summarizes these relations.

\begin{figure}[t]
 \centering
  \resizebox{0.4\linewidth}{!}{%
  \begingroup
\fontencoding{OT1}\selectfont
\definecolor{ClassResult}{RGB}{166,48,50}
\begin{tikzpicture}[
  class/.style={font=\normalsize,inner sep=4pt},
  inclusion/.style={-{Stealth[length=2mm]},draw=black!60,line width=0.7pt},
  new strictness/.style={inclusion,draw=ClassResult,line width=1pt},
  relation/.style={font=\small,inner sep=2pt}
]
  \node[class] (upper) at (0,1.6) {$\Sigma_2\cap\Pi_2$};
  \node[class] (bpp) at (-3.2,0) {$\BPP$};
  \node[class] (pnp) at (3.2,0) {$\PNP$};
  \node[class] (prp) at (0,-1.6) {$\PRP$};

  \draw[inclusion] (bpp) -- (upper);
  \draw[inclusion] (pnp) -- (upper);
  \draw[inclusion] (prp) -- (pnp);
  \draw[new strictness] (prp) --
    node[midway,sloped,below=3pt,font=\footnotesize,text=ClassResult]
      {strict (this work)} (bpp);

  \draw[densely dotted,draw=black!55,line width=0.7pt]
    (bpp.east) -- (pnp.west);

  \node[relation,text=ClassResult] at (0,0.28)
    {$\BPP\not\subseteq\PNP$\quad (this work)};
  \node[relation,text=black!65] at (0,-0.28)
    {$\PNP\not\subseteq\BPP$\quad (known)};
\end{tikzpicture}
\endgroup}
 \caption{\small Selected communication classes on total functions;
 new separations are highlighted in red.}
 \label{fig:class-relations}
\end{figure}
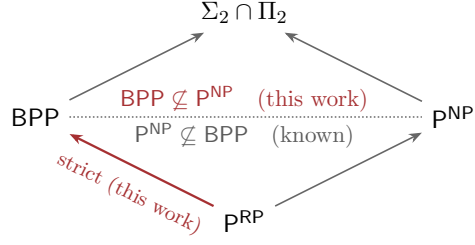

\subsection{Proof overview}
Our construction uses the cheat-sheet method of
Aaronson, Ben-David, and Kothari~\cite{aaronson2016separations},
introduced in query complexity to obtain separations for total
functions from promise problems. Anshu
et al.~\cite{anshu2016separations} developed the method in communication
complexity.

We use Gap-Hamming with a linear gap as the outer gadget: its outputs
address a table of certificates, the cheat-sheet. The inner gadget uses fully linear
PCPs\footnotemark~to certify that the base input satisfies the Gap-Hamming domain promises.
Gavinsky~\cite{gavinsky2026quantum} recently used this inner gadget to
separate quantum and randomized communication for a total function.

\footnotetext{We remark here, the fully linear PCP of Boneh et al. is designed for secret-shared inputs. Its verifier accesses the input and proof jointly through a constant number of linear functionals; Boneh et al. call such a system a fully linear PCP. It has no implication to the proof length.}

At a high level, the randomized upper bound is from random sampling for
Gap-Hamming and the verification guarantee of fully linear PCPs.
The remaining task is to show that the resulting total function
has no large monochromatic rectangle.

Fix an integer $k$, let $m=2^k$, and put $d=O(km^2)$.
Alice and Bob receive inputs
\[
 (\bx,\bu),\ (\by,\bv)\in
 \bits^{k\cdot m}\times\bits^{m\cdot 4k(2d-1)},
\]
respectively. The strings $\bx$ and $\by$ are the inputs to $k$
Gap-Hamming instances, each on $m$ bits per party. The strings
$\bu$ and $\bv$ are the cheat-sheet inputs: each is a table of $m$
entries, with $4k(2d-1)$ bits per entry. Thus each party's input has length $n=km+4km(2d-1)$, so
$k=\Theta(\log n)$ and $m=n^{\Theta(1)}$. This construction is illustrated in~\cref{fig:cheat-sheet}.

\paragraph{Gap-Hamming supplies an address.}
Write $\bx=(\bx_1,\ldots,\bx_k)$ and
$\by=(\by_1,\ldots,\by_k)$, where
$\bx_i,\by_i\in\bits^m$, and put $\bz=\bx\oplus\by$.
For each pair, $\GH(\bx_i,\by_i)$ is $0$ when $|\bz_i|\le m/3$,
is $1$ when $|\bz_i|\ge 2m/3$, and is undefined otherwise.
The thresholds are separated by $m/3$.
Equivalently, each instance applies gap majority to the XOR block
$\bz_i$.
The indicator $\indicate(\bz)$ records whether all $k$ pairs satisfy
these promises. When $\indicate(\bz)=1$, their $k$ answers specify
an address $\adrs(\bz)\in[m]$, using the encoding in
\cref{sec:construction}.

The linear gap makes this address inexpensive to estimate. To make
all $k$ address bits correct with constant probability, we reduce the
error on each instance to $O(1/k)$. This requires $O(\log k)$ sampled
coordinates per instance and hence $O(k\log k)$ communication.

\paragraph{Fully linear PCPs supply the certificates.}
We use the fully linear PCP construction of
Boneh et al.~\cite{boneh2019zero} to certify that all Gap-Hamming
instances satisfy their promises. For each $\bz$, a string
$\bc\in\bits^{4k(2d-1)}$ is a proposed certificate for
$\indicate(\bz)=1$. We define the certificate predicate $\cfunc(\bz,\bc)$
so that an accepting certificate exists if and only if $\indicate(\bz)=1$.

A fully linear PCP is a
probabilistically checkable proof whose verifier accesses both the
input and the proof through linear functionals over a finite field.
The construction we use requires only $O(1)$ such queries.

This query model is particularly useful when the input and proof are
shared between Alice and Bob. Suppose
$\bc=\bc_A\oplus\bc_B$. Over $\F_{2^{4k}}$, the XOR shares become
additive shares, so for any linear functional $L$,
\[
 L(\bz,\bc)
 =L(\bx,\bc_A)+L(\by,\bc_B).
\]
The public randomness specifies all the queries. For each one, Alice
evaluates the functional on her shares, and Bob can recover the answer
by adding his local value. Since the queries are nonadaptive, Alice
sends all $O(1)$ field elements in one message.

\paragraph{The combined cheat sheet.}
The cheat-sheet table has one cell for each of the $m=2^k$ possible
vectors of Gap-Hamming outputs. Alice holds
$\bu=(\bu_1,\ldots,\bu_m)$ and Bob holds
$\bv=(\bv_1,\ldots,\bv_m)$. At entry $j$, their two shares specify
the proposed certificate
\[
 \bc_j=\bu_j\oplus\bv_j.
\]
When the base input satisfies every promise, its $k$ Gap-Hamming
outputs select one cell, and the function checks only the certificate
in that cell. Writing $h=\adrs(\bz)$, we define
\[
 \tGH(\bx,\bu;\by,\bv)
 =
 \begin{cases}
 \cfunc(\bz,\bc_h),
     & \text{if }\indicate(\bz)=1,\\
 0,  & \text{if }\indicate(\bz)=0.
 \end{cases}
\]
This definition covers every input. If some Gap-Hamming instance
violates its promise, the output is $0$ independently of the table.
If all instances satisfy their promises, every cell except the
$\ell$-th is ignored, and the output is $1$ exactly when
$\bc_\ell$ is the unique valid certificate. The table is unrestricted,
so an invalid certificate at the selected cell produces output $0$.
Figure~\ref{fig:cheat-sheet} depicts the combined construction.

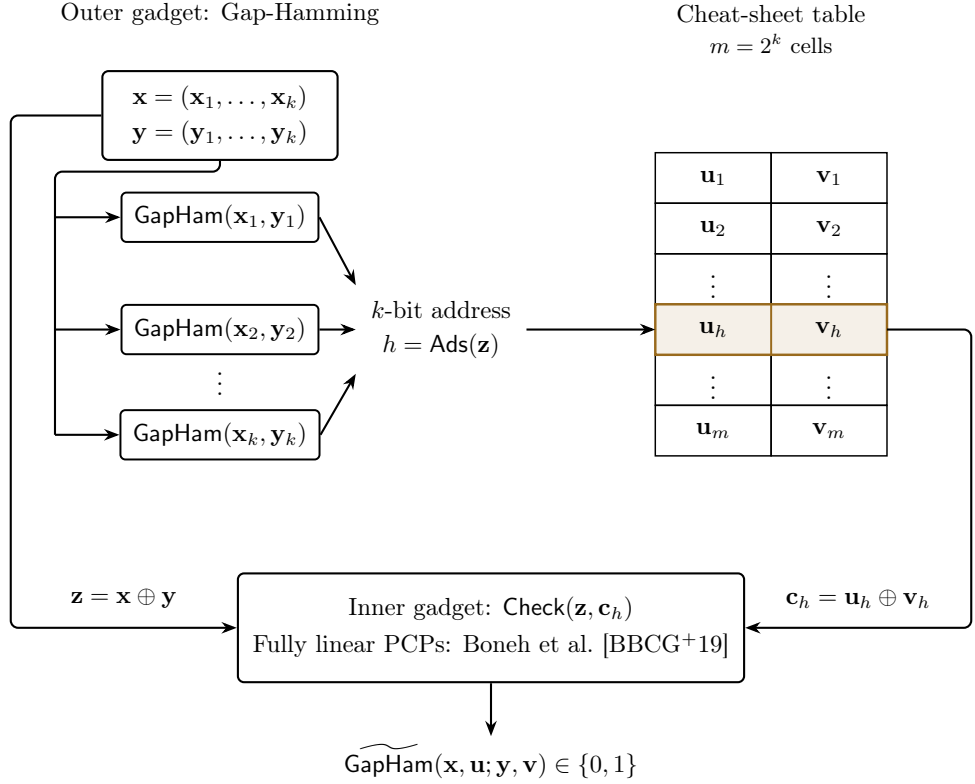
\begin{figure}[t]
 \centering
   \resizebox{0.8\linewidth}{!}{%
    \begingroup
\definecolor{CheatFrame}{RGB}{0,0,0}
\definecolor{CheatCell}{RGB}{153,106,29}
\begin{tikzpicture}[
  font=\small,text=black,
  flow/.style={-{Stealth[length=2.2mm,width=1.65mm]},
    draw=black,line width=0.85pt,rounded corners=3pt},
  gadget/.style={draw=CheatFrame,fill=white,
    line width=0.8pt,rounded corners=2pt,
    minimum width=2.8cm,minimum height=0.68cm,inner sep=5pt},
  cell/.style={draw=CheatFrame,line width=0.6pt,
    minimum width=1.65cm,minimum height=0.72cm,inner sep=2pt}
]
  \node at (2.8,7.85)
    {Outer gadget: Gap-Hamming};
  \node[draw=CheatFrame,line width=0.85pt,rounded corners=2.5pt,
    fill=white,align=center,minimum width=3.35cm,
    minimum height=0.92cm,inner sep=6pt] (base) at (2.8,6.40)
    {$\bx=(\bx_1,\ldots,\bx_k)$\\[2pt]
     $\by=(\by_1,\ldots,\by_k)$};

  \node[gadget] (gh1) at (2.8,4.95) {$\GH(\bx_1,\by_1)$};
  \node[gadget] (gh2) at (2.8,3.35) {$\GH(\bx_2,\by_2)$};
  \node at (2.8,2.68) {$\vdots$};
  \node[gadget] (ghk) at (2.8,1.85) {$\GH(\bx_k,\by_k)$};

  \draw[draw=black,line width=0.85pt,rounded corners=3pt]
    (base.south) -- (2.8,5.60) -- (0.45,5.60) -- (0.45,1.85);
  \draw[flow] (0.45,4.95) -- (gh1.west);
  \draw[flow] (0.45,3.35) -- (gh2.west);
  \draw[flow] (0.45,1.85) -- (ghk.west);

  \node[align=center,inner sep=6pt] (address) at (5.95,3.35)
    {$k$-bit address\\[2pt]$h=\adrs(\bz)$};
  \draw[flow] (gh1.east) -- (address.north west);
  \draw[flow] (gh2.east) -- (address.west);
  \draw[flow] (ghk.east) -- (address.south west);

  \node at (10.65,7.85) {Cheat-sheet table};
  \node[font=\footnotesize] at (10.65,7.43) {$m=2^k$ cells};

  \foreach \pos/\idx in {5.51/1,4.79/2,1.91/m} {
    \node[cell] at (9.825,\pos) {$\bu_{\idx}$};
    \node[cell] at (11.475,\pos) {$\bv_{\idx}$};
  }
  \foreach \pos in {4.07,2.63} {
    \node[cell] at (9.825,\pos) {$\vdots$};
    \node[cell] at (11.475,\pos) {$\vdots$};
  }

  \draw[draw=CheatCell,fill=CheatCell!10,line width=1pt]
    (9.0,2.99) rectangle (12.3,3.71);
  \draw[draw=CheatCell,line width=0.8pt]
    (10.65,2.99) -- (10.65,3.71);
  \node at (9.825,3.35) {$\bu_h$};
  \node at (11.475,3.35) {$\bv_h$};

  \draw[flow] (address.east) --
    node[above=3pt,font=\footnotesize] {}
    (9.0,3.35);

  \node[draw=CheatFrame,fill=white,line width=0.85pt,
    rounded corners=2.5pt,align=center,minimum width=5.55cm,
    minimum height=1.55cm,inner sep=6pt] (check) at (6.66,-0.90)
    { Inner gadget: $\cfunc(\bz,\bc_h)$\\[3pt]
     Fully linear PCPs:
     Boneh et al.~\cite{boneh2019zero}};

  \draw[flow] (12.3,3.35) -- (13.50,3.35) --
    (13.50,-0.90) --
    node[above=4pt,midway] {$\bc_h=\bu_h\oplus\bv_h$}
    (check.east);

  \draw[flow] (base.west) -- (-0.18,6.40) -- (-0.18,-0.90) --
    node[above=4pt,midway] {$\bz=\bx\oplus\by$}
    (check.west);

  \node[inner sep=3pt] (output) at (6.66,-2.80)
    {$\tGH(\bx,\bu;\by,\bv)\in\bits$};
  \draw[flow] (check.south) -- (output.north);
\end{tikzpicture}
\endgroup
  }
 \caption{\small The Gap-Hamming cheat-sheet construction on promised base
 inputs. The $k$ Gap-Hamming outputs select one of $m=2^k$ entries in
 the XOR-shared certificate table. The selected certificate is checked
 against the base input using fully linear PCPs of
 Boneh et al.~\cite{boneh2019zero}; the inner predicate tests whether
 this certificate proves $\indicate(\bz)=1$. }
 \label{fig:cheat-sheet}
\end{figure}

\paragraph{Excluding large monochromatic rectangles.}
Ruling out a large $0$-rectangle is delicate and involves several
steps. Given a somewhat large rectangle, we find a $1$-input inside it as follows.
First, we retain base inputs on each side that have many compatible
tables; the resulting sets of base inputs remain large. Second, a
counting argument shows that each retained base input has
only a few addresses at which at most half of all table-entry values
occur. Third, double counting gives an address $h$ for which more than
half of all values occur for many retained inputs on both sides. We
restrict to these inputs, obtaining two large sets $X_h$ and $Y_h$.
Fourth, the Hamming-distance argument finds $\bx\in X_h$ and
$\by\in Y_h$ whose Gap-Hamming outputs select address $h$. Finally, let
$\bc^*$ be the valid certificate for this pair. Because $\bx\in X_h$,
more than half of all certificate shares occur as the $h$-th entry of
a table compatible with $\bx$ on Alice's side; the analogous statement
holds for $\by\in Y_h$ on Bob's side. After translating Bob's set of
shares by $\bc^*$, these two sets still contain more than half of all
shares and therefore intersect. This gives compatible tables whose
$h$-th entries have XOR $\bc^*$, producing a $1$-input in the proposed
$0$-rectangle. Large $1$-rectangles are ruled out directly by the
density of $1$-inputs.

\section{Preliminaries}
\label{sec:preliminaries}

\paragraph{Notation.}
For a positive integer $n$, write $[n]=\{1,\ldots,n\}$.
All logarithms are to base $2$, except that $\ln$ denotes the natural
logarithm. We write $\polylog n$ for $(\log n)^{O(1)}$.
For a binary string $\bz$, write $\bz(j)$ for its $j$-th
coordinate and $|\bz|$ for its Hamming weight. When a string is
partitioned into blocks, subscripts index the blocks, so $\bz_i(j)$
is the $j$-th coordinate of block $\bz_i$. The operation $\oplus$
denotes bitwise XOR, so $|\bx\oplus\by|$ is the Hamming distance. For an event $E$, its indicator
$\ind\{E\}$ is $1$ when $E$ holds and $0$ otherwise.
The density of a subset $S$ of a finite nonempty set $\Omega$ is
$|S|/|\Omega|$.

\paragraph{Rectangles.}
Let $F:\mathcal X\times\mathcal Y\to\bits$ be a total Boolean
function on finite nonempty input sets. A rectangle is a set
$A\times B$ with $A\subseteq\mathcal X$ and $B\subseteq\mathcal Y$.
It is a monochromatic $b$-rectangle if $F(x,y)=b$ for every $(x,y)\in A\times B$, $b\in \{0,1\}$.
We write
\[
 \rect(F)=\max\left\{
 \frac{|A||B|}{|\mathcal X||\mathcal Y|}:
 A\times B\text{ is monochromatic for }F
 \right\}.
\]
Thus $\rect(F)$ is the largest monochromatic rectangle density
under the uniform distribution.
For a product probability measure $\mu=\mu_{\mathcal X}\times
\mu_{\mathcal Y}$, we have  $\mu(A\times B): =  \mu_{\mathcal X}(A)\mu_{\mathcal Y}(B)$.

\paragraph{Communication models.}
Alice receives $x\in\mathcal X$ and Bob receives $y\in\mathcal Y$. Write $\Rcc(F)$ for the minimum
communication cost of a public-coin randomized protocol computing
$F$ with error at most $1/3$ on every input. Its cost is the worst-case communication cost over all inputs and random bits.

For $\PNPcc(F)$, we use the model of
G\"o\"os et al.~\cite{goos2017query}.
A protocol is a deterministic binary decision tree whose leaves
are labeled by output bits. At each internal node, the parties
receive the exact value of a predicate $Q(x,y)$. A query has cost
$h\ge1$ if its $1$-set can be covered by at most $2^h$ rectangles.
The answer determines the next node, so queries may depend on
earlier answers. The cost of the protocol is the maximum, over all
inputs, of the sum of query costs along the resulting path.
We write $\PNPcc(F)$ for the minimum cost of a protocol computing
$F$ exactly. Ordinary communication is included: sending one bit
is a query of cost $1$. Figure~\ref{fig:np-query-protocol} illustrates
the rectangle covers and adaptive queries in this model.

Following Hatami and Hatami~\cite{hatami2026structure},
define $\PRPcc(F)$ in the same way, with each query charged its
public-coin one-sided randomized communication complexity.
Such a randomized protocol must output $0$ whenever $Q(x,y)=0$
and must output $1$ with probability at least $2/3$ whenever
$Q(x,y)=1$. We impose a minimum query cost of $1$, as in the
$\NP$-query model. The oracle still returns $Q(x,y)$ exactly;
randomization enters only in the cost assigned to the query.

For families $\{F_n:\bits^n\times\bits^n\to\bits\}$, the classes
$\BPP$, $\PNP$, and $\PRP$ consist of those with respective costs
$\Rcc(F_n)$, $\PNPcc(F_n)$, and $\PRPcc(F_n)$ bounded by
$(\log n)^{O(1)}$. All complexity classes in this paper refer to
communication complexity.

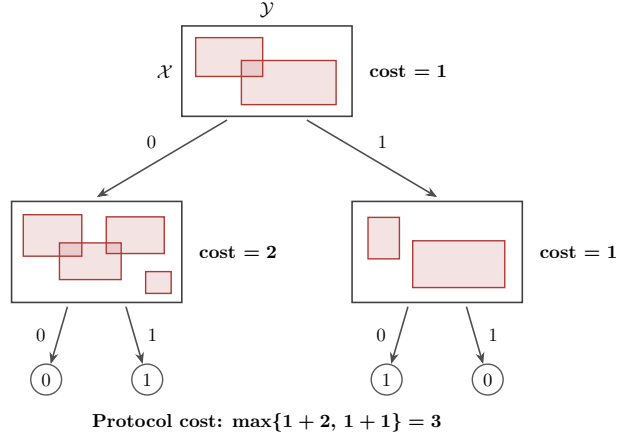
\begin{figure}[tb]
 \centering
 \resizebox{0.5\linewidth}{!}{%
 \begingroup
\definecolor{QueryRed}{RGB}{173,59,55}
\begin{tikzpicture}[
  font=\small,
  cover/.style={draw=QueryRed,line width=0.75pt,
    fill=QueryRed,fill opacity=0.14},
  frame/.style={draw=black!75,line width=0.85pt},
  branch/.style={-{Stealth[length=2.2mm,width=1.7mm]},
    draw=black!70,line width=0.85pt},
  answer/.style={inner sep=1pt},
  cost/.style={font=\small\bfseries\boldmath,inner sep=2pt},
  output/.style={circle,draw=black!65,line width=0.75pt,
    fill=white,minimum size=5.8mm,inner sep=1pt}
]
  \node at (0,2.00) {$\mathcal Y$};
  \node at (-1.92,0.85) {$\mathcal X$};
  \draw[frame] (-1.6,0) rectangle (1.6,1.7);
  \draw[cover] (-1.34,0.75) rectangle (-0.08,1.48);
  \draw[cover] (-0.48,0.22) rectangle (1.30,1.05);
  \node[cost,anchor=west] at (1.85,0.85) {cost $=1$};

  \draw[branch] (-0.75,-0.07) -- (-3.20,-1.52);
  \draw[branch] (0.75,-0.07) -- (3.20,-1.52);
  \node[answer] at (-2.17,-0.49) {$0$};
  \node[answer] at (2.17,-0.49) {$1$};

  \begin{scope}[shift={(-3.20,-3.50)}]
    \draw[frame] (-1.6,0) rectangle (1.6,1.90);
    \draw[cover] (-1.38,0.87) rectangle (-0.28,1.65);
    \draw[cover] (-0.70,0.43) rectangle (0.46,1.12);
    \draw[cover] (0.18,0.92) rectangle (1.27,1.61);
    \draw[cover] (0.92,0.17) rectangle (1.40,0.58);
  \end{scope}
  \node[cost,anchor=west] at (-1.35,-2.55) {cost $=2$};

  \begin{scope}[shift={(3.20,-3.50)}]
    \draw[frame] (-1.6,0) rectangle (1.6,1.90);
    \draw[cover] (-1.30,0.82) rectangle (-0.71,1.60);
    \draw[cover] (-0.46,0.28) rectangle (1.27,1.16);
  \end{scope}
  \node[cost,anchor=west] at (5.05,-2.55) {cost $=1$};

  \node[output] (out00) at (-4.15,-4.95) {$0$};
  \node[output] (out01) at (-2.25,-4.95) {$1$};
  \node[output] (out10) at (2.25,-4.95) {$1$};
  \node[output] (out11) at (4.15,-4.95) {$0$};

  \draw[branch] (-3.75,-3.58) -- (out00);
  \draw[branch] (-2.65,-3.58) -- (out01);
  \draw[branch] (2.65,-3.58) -- (out10);
  \draw[branch] (3.75,-3.58) -- (out11);

  \node[answer] at (-4.25,-4.12) {$0$};
  \node[answer] at (-2.15,-4.12) {$1$};
  \node[answer] at (2.15,-4.12) {$0$};
  \node[answer] at (4.25,-4.12) {$1$};

  \node[cost] at (0,-5.75)
    {Protocol cost: $\max\{1+2,\,1+1\}=3$};
\end{tikzpicture}
\endgroup}
 \caption{\small An adaptive $\PNP$ protocol in the model of
 G\"o\"os et al.~\cite{goos2017query}.
 Each box depicts a query on $\mathcal X\times\mathcal Y$:
 the shaded rectangles cover its $1$-set, and the white region is
 its $0$-set. A cover by at most $2^h$ rectangles permits a
 query of cost $h\ge1$. The exact answer selects the next query;
 leaves give the final output.}
 \label{fig:np-query-protocol}
\end{figure}

\paragraph{The product method.}
We use the following lower bound of Impagliazzo and
Williams~\cite{impagliazzo2010communication}; see also
\cite{goos2017query}.
\begin{fact}[Impagliazzo-Williams]\label{fact:product-method}
Let $\mu$ be a product probability measure on
$\mathcal X\times\mathcal Y$. If every monochromatic rectangle
of $F$ has $\mu$-measure at most $\delta$, where $0<\delta\le1$,
then $\PNPcc(F)=\Omega(\log(1/\delta))$. In particular,
\[
 \PNPcc(F)=\Omega\!\left(\log\frac1{\rect(F)}\right).
\]
\end{fact}

\section{Totalization via cheat sheets}
\label{sec:construction}

We extend the Gap-Hamming communication problem to a total function by adding tables of certificates, following the cheat-sheet construction~\cite{anshu2016separations,aaronson2016separations}.
Throughout the construction, fix an integer $k\ge 12$ and let $m=2^k$.

\begin{definition}
For $\ba,\bb\in\bits^m$, the Gap-Hamming function is
\[
\GH(\ba,\bb)=
\begin{cases}
0, & |\ba\oplus\bb|\le m/3,\\
1, & |\ba\oplus\bb|\ge 2m/3,\\
*, & m/3<|\ba\oplus\bb|<2m/3.
\end{cases}
\]
The symbol $*$ means that the output is undefined.
\end{definition}

Define
\[
\bx=(\bx_1,\ldots,\bx_k),\qquad
\by=(\by_1,\ldots,\by_k),\qquad
\bz=(\bz_1,\bz_2,...,\bz_k )  = \bx\oplus\by,
\]
where $\bx_i,\by_i\in\bits^m$ and $\bz_i=\bx_i\oplus\by_i$.
For each   $h\in[m]$, let $(h_1,h_2,\ldots,h_k)\in\bits^{k}$ be the  binary expansion, i.e.,
\[
h=1+\sum_{i=1}^k h_i2^{k-i}.
\]
Thus the outputs of $k$ Gap-Hamming instances specify an address in the table.
\begin{definition}
For $\bz \in\bits^{km}$, define the indicator function as
\[
\indicate(\bz)=
\prod_{i=1}^k
\ind\big\{|\bz_i|\le m/3\text{ or }|\bz_i|\ge 2m/3\big\}.
\]
If $\GH$ is well defined on  all $(\bx_i,\by_i)$, i.e., $\indicate(\bz)=1$,
we obtain an address
\[
\adrs(\bz) :=
\displaystyle 1+\sum_{i=1}^k
\GH(\bx_i,\by_i)\cdot 2^{k-i} \in [m].
\]
Otherwise, define $\adrs(\bz)$ arbitrarily in $[m]$.
\end{definition}

\subsection{A circuit for \texorpdfstring{$1-\indicate(\bz)$}{1-Ind(z)}.}

In the cheat sheet totalization construction, the fully linear PCP from Boneh et al.~\cite{boneh2019zero} serves to checks whether all Gap-Hamming inputs satisfy their promises. The PCP builds on a circuit with certificate; the circuit computes $1-\indicate(\bz)$; and the certificate claims the circuit will output $0$, which is equivalent to $\indicate(\bz)=1$. Alice and Bob can efficiently check this claim by exchanging  $O(k)$ bits. We first describe the circuit.

\begin{lemma}\label{lem:circuit}
There is a circuit using fan-in two  $\mathsf{XOR}$
and $\andgadget$ gates, $\mathsf{NOT}$ gates, and constants $\{0,1\}$, that computes
$1-\indicate(\bz)$ for every $\bz\in\bits^{km}$.
Furthermore, it includes
$
d=k\big(m(m+3)+1\big)
$
$\andgadget$ gates.
\end{lemma}

For each $i\in[k]$ and $0\le j\le m$, define $b^{(i)}_{j,r}=0$ whenever $r<0$ or $r>j$, and set $b^{(i)}_{0,0}=1$.
For $1\le j\le m$ and $0\le r\le j$, given the input $\bz$, compute
\begin{equation}\label{eq:1}
b_{j,r}^{(i)}
=\big(\lnot\bz_i(j)\land b_{j-1,r}^{(i)}\big)
 \oplus\big(\bz_i(j)\land b_{j-1,r-1}^{(i)}\big).
\end{equation}
The following identity shows that these bits record the number of ones in each prefix segment of $\bz_i$.

\begin{claim}\label{fact:counting}
For each $i\in[k]$, $0\le j\le m$, and $0\le r\le j$,
\[
b_{j,r}^{(i)}
=\ind\Big\{\sum_{t=1}^j\bz_i(t)=r\Big\}.
\]
\end{claim}

\begin{proof}
We prove the identity by induction on $j$.
For $j=0$, the sum is empty and $b^{(i)}_{0,0}=1$.
Suppose the identity holds for $j-1$.
By the zero values fixed above, it also holds for $r<0$ or $r>j-1$.

If $\bz_i(j)=0$, then \eqref{eq:1} gives
\[
\begin{aligned}
b_{j,r}^{(i)}
=b_{j-1,r}^{(i)}
=\ind\Big\{\sum_{t=1}^{j-1}\bz_i(t)=r\Big\}
=\ind\Big\{\sum_{t=1}^{j}\bz_i(t)=r\Big\}.
\end{aligned}
\]
If $\bz_i(j)=1$, then \eqref{eq:1} gives
\[
\begin{aligned}
b_{j,r}^{(i)}
=b_{j-1,r-1}^{(i)}
=\ind\Big\{\sum_{t=1}^{j-1}\bz_i(t)=r-1\Big\}
=\ind\Big\{\sum_{t=1}^{j}\bz_i(t)=r\Big\}.
\end{aligned}
\]
These two cases prove the identity for $j$.
\end{proof}

We describe the circuit as follows.
\begin{proof}[Proof of Lemma~\ref{lem:circuit}]
    Given the $\{b_{j,r}^{(i) } \}$,
for each $i\in[k]$, the circuit computes
\[
B_i=
\bigoplus_{\substack{0\le r\le m\\r\le m/3\ \text{or}\ r\ge 2m/3}}
b_{m,r}^{(i)}.
\]
By Claim~\ref{fact:counting}, exactly one of $\{b_{m,0}^{(i)},\ldots,b_{m,m}^{(i)} \}$ equals $1$.
Therefore
$
B_i=\ind\big\{|\bz_i|\le m/3\text{ or }|\bz_i|\ge 2m/3\big\}.
$
Then the circuit outputs
\begin{equation}\label{eq:2}
\lnot(B_1\land\cdots\land B_k)\land 1
=1-\prod_{i=1}^k B_i
=1-\indicate(\bz).
\end{equation}
Specifically, compute $\lnot(B_1\land\cdots\land B_k)\land 1$ using $\andgadget$ gates with two inputs.

Each pair $(j,r)$ in \eqref{eq:1} uses two $\andgadget$ gates, and each $j$ has $j+1$ choices of $r$.
Equation~\eqref{eq:2} uses $k-1$ gates for $B_1\land\cdots\land B_k$ and one final gate. So the number of $\andgadget$ gates is
\[
2k\sum_{j=1}^{2^k}(j+1)+k
=k\big(2^k(2^k+3)+1\big) = O(km^{2}).\qedhere
\]

\end{proof}

In the next section, we introduce the certificate.

\subsection{Circuits to certificates.}

We arithmetize the circuit for use in the fully linear PCP.
In particular, we will express operations of the circuit over a finite field and encode
its $\andgadget$ outputs by a polynomial.

At first, label the $\andgadget$ gates according to the calculation order (bottom-up) by  $1,\ldots,d$. Recall that $d= k\big(m(m+3)+1\big)$ by Lemma~\ref{lem:circuit}. We perform arithmetic in the field $\F_{2^{4k}}$  which has characteristic 2. More specifically,
fix $e_1=1,e_2,\ldots,e_{4k}\in\F_{2^{4k}}$ such that every
$a\in\F_{2^{4k}}$ can be written uniquely as
\[
a=\sum_{i=1}^{4k}a_i e_i,\qquad a_i\in\{0,1\}.
\]
These elements form a basis over $\F_2=\{0,1\}$, and the
coefficient string $(a_1,\ldots,a_{4k})$ encodes $a$.
In particular, the field elements $0$ and $1$ are encoded by
$(0,\ldots,0)$ and $(1,0,\ldots,0)$. Addition in $\F_{2^{4k}}$ corresponds to bitwise XOR of these
$4k$-bit encodings: each coordinate is added modulo $2$,
while the result remains a $4k$-bit encoding of a field element.
We represent Boolean inputs $0,1$ by these field elements.
For Boolean $a,b$, the gate operations therefore satisfy
\[
a\land b=ab,\qquad a\oplus b=a+b,\qquad \lnot a=1+a,
\]
where the operations on the right are field operations. Since $d<2k\,4^k<2^{4k}$, we can fix distinct
$$\tau_1,\tau_2,\ldots,\tau_d\in\F_{2^{4k}}.$$

\begin{definition}[Certificate]
   A certificate $\bc\in\bits^{4k\cdot (2d-1)}$ consists of  $2d-1$ blocks of $4k$ bits.
Let $c_0,\ldots,c_{2d-2}\in\F_{2^{4k}}$ be the field elements encoded by these blocks, and define the polynomial in the variable $T$ by
\[
p(T)=\sum_{v=0}^{2d-2}c_vT^v.
\]
\end{definition}

Certificates and polynomials of degree at most $2d-2$ correspond one to one.
The certificate claims that the $v$-th $\andgadget$ gate  outputs $p(\tau_v)$;
and we show how to check this claim efficiently in the next section.

\subsection{The fully linear PCP}

Given $\bz$ and the polynomial $p$ encoded by certificate $\bc$,
process the circuit in calculation order.
At an $\mathsf{XOR}$ or $\mathsf{NOT}$ gate, apply the field operations above.
At the $v$-th $\andgadget$ gate, use $p(\tau_v)$ directly as its output.

Let $a_v(\bz,p)$ and $b_v(\bz,p)$ be the left and right inputs obtained at the $v$-th $\andgadget$ gate.
They depend only on $\bz$ and the earlier claimed outputs $p(\tau_1),\ldots,p(\tau_{v-1})$.
Each is a fixed weighted sum of the coordinates of $\bz$ and the coefficients of $p$, plus a constant: the only operations used between claimed outputs are addition and addition of $1$.
So $a_v(\bz,p)$ and $b_v(\bz,p)$ are
{affine} functions on $\bz$ and $p(\tau_1),\ldots,p(\tau_{d})$.

For $v\in[d]$, define
\[
\delta_v(T)=\prod_{\substack{j\in[d]\\j\ne v}}
\frac{T-\tau_j}{\tau_v-\tau_j}.
\]
These polynomials satisfy $\delta_v(\tau_v)=1$ and $\delta_v(\tau_j)=0$ for $j\ne v$.
Define
\[
f_{(\bz,p)}(T)=\sum_{v=1}^d\delta_v(T)a_v(\bz,p),\quad
g_{(\bz,p)}(T)=\sum_{v=1}^d\delta_v(T)b_v(\bz,p).
\]
Their degrees are at most $d-1$, and
\[
f_{(\bz,p)}(\tau_v)=a_v(\bz,p),\qquad
g_{(\bz,p)}(\tau_v)=b_v(\bz,p).
\]
For each fixed $T$, both evaluations are affine functions of $\bz$ and the coefficients of $p$.

We require a polynomial identity to enforce the multiplication equations, and require the final output to be $0$.
Define
\[
\cfunc(\bz,p)=
\begin{cases}
1, & p\equiv f_{(\bz,p)}g_{(\bz,p)}\text{ and }p(\tau_d)=0,\\
0, & \text{otherwise}.
\end{cases}
\]
Let $\cfunc(\bz,\bc) :=\cfunc(\bz,p)$ for the polynomial $p$ encoded by $\bc$.

The following lemma shows that a true  certificate exists exactly when $\indicate(\bz)=1$, and that the polynomial identity makes the certificate unique. For brevity, we defer the proof to  Appendix~\ref{app:auxiliary-proofs}.

\begin{lemma}\label{lem:check_func}
The function $\cfunc(\bz,p)$ determines whether $p$ certifies that the circuit output is $0$.
Moreover,
\begin{enumerate}
\item[1).] If $\indicate(\bz)=1$, exactly one certificate $\bc\in\bits^{4k(2d-1)}$ satisfies $\cfunc(\bz,\bc)=1$.
\item[2).] If $\indicate(\bz)=0$, every certificate $\bc\in\bits^{4k(2d-1)}$ satisfies $\cfunc(\bz,\bc)=0$.
\end{enumerate}
\end{lemma}

\medskip

We show the randomized test for the polynomial  condition in $\cfunc$ as follows.
A linear query returns a weighted sum of the coordinates
of $\bz$ and the coefficients of $p$. Let $\mathbf{0}\in\bits^{km}$ be the zero vector,
and let $p_0(T)\equiv0$ be the zero polynomial.

\begin{theorem}
[Fully linear PCP, cf.~\cite{boneh2019zero}] \label{lem:linear-check}
Fix $\bz\in\bits^{km}$ and a polynomial
$p\in\F_{2^{4k}}[T]$ of degree at most $2d-2$.
There is a randomized test using four linear queries: if $\cfunc(\bz,p)=1$, it always outputs $1$; if $\cfunc(\bz,p)=0$,
it outputs $1$ with probability
at most $1/6$.
\end{theorem}

\begin{proof}
Choose $r\in\F_{2^{4k}}$ uniformly.
The four queries return
\[
p(r),\quad
f_{(\bz,p)}(r)-f_{(\mathbf{0},p_0)}(r),\quad
g_{(\bz,p)}(r)-g_{(\mathbf{0},p_0)}(r),\quad
p(\tau_d).
\]
These are linear functions of $\bz$ and the coefficients
of $p$: the second and third expressions subtract
the constants from the corresponding affine functions, which depend only on the circuit and $r$. Adding  constants $f_{(\mathbf{0},p_0)}(r)$ and $g_{(\mathbf{0},p_0)}(r)$
to the second and third answers gives $ f_{(\bz,p)}(r) $ and $ g_{(\bz,p)}(r)$, respectively.

The test outputs $1$ if and only if
\[
p(\tau_d)=0
\quad\text{and}\quad
p(r)=f_{(\bz,p)}(r)g_{(\bz,p)}(r).
\]

If $\cfunc(\bz,p)=1$, both equalities hold for every $r$. Suppose $\cfunc(\bz,p)=0$.
If $p(\tau_d)\ne0$, the test always outputs $0$.
Otherwise, $p-f_{(\bz,p)}g_{(\bz,p)}$ is a nonzero polynomial
of degree at most $2d-2$.
It has at most $2d-2$ roots, so
\[
\Pr_r[\text{the test outputs }1]
\le\frac{2d-2}{2^{4k}}<\frac16.\qedhere
\]
\end{proof}

\subsection{The total function}
Now we combine the fully linear PCP and Gap-Hamming function to construct our total function.
Besides $\bx\in \{0,1\}^{km}$ and $\by\in \{0,1\}^{km}$,
Alice and Bob each receive a table with $m$ positions:
\[
\bu=(\bu_1,\ldots,\bu_m),\qquad
\bv=(\bv_1,\ldots,\bv_m),
\]
where $\bu_j,\bv_j\in\bits^{4k(2d-1)}$.
At position $j$, the two table entries together specify the certificate
\[
\bc_j=\bu_j\oplus\bv_j.
\]
The total function uses the certificate at the address determined by the Gap-Hamming outputs.

\begin{definition}[Gap-Hamming with cheat sheets]
\label{def:total_function}
For inputs $(\bx,\bu)$ and $(\by,\bv)$ in $\bits^{km}\times\bits^{4km(2d-1)}$, let $\bz=\bx\oplus\by$ and $h=\adrs(\bz)\in[m]$.
Define
\[
\tGH(\bx,\bu;\by,\bv)=
\begin{cases}
\cfunc(\bz,\bc_h), & \indicate(\bz)=1,\\
0, & \indicate(\bz)=0.
\end{cases}
\]
\end{definition}

Each player  receives an input of length
$
n=km+4km(2d-1).
$
Substituting the value of $d$ gives
$
k=\Theta(\log n).
$
For the two-party communication model, Alice has $(\bx,\bu)$ and Bob has $(\by, \bv)$;
we give a randomized protocol on $\tGH$ with cost $O(\log n\log\log n)$ in the next section.

\section{An Efficient Randomized Protocol}
\label{sec:randomized-protocol}

To compute $\tGH$, Alice and Bob first estimate an address by a random sketch of inputs $(\bx,\by)$.
Then they apply the test from \cref{lem:linear-check} to estimate the $\cfunc$ function.

\begin{theorem}
\label{thm:randomized_protocol}
There is a two-party communication protocol using public randomness that computes $\tGH$ with error probability less than $1/3$ and total communication $O(k\log k)$.
\end{theorem}

\begin{proof}
\textbf{Protocol.}
Alice holds $(\bx,\bu)$, and Bob holds $(\by,\bv)$.
\begin{enumerate}[label=\Roman*).]
\item
Let $q=\lceil18\ln(6k)\rceil$.
Using public randomness, the parties sample positions $I_{i,j}$ uniformly and independently from $[m]$, for $i\in[k]$ and $j\in[q]$.
Positions may be sampled more than once.
\item
Alice sends the $kq$ bits $\bx_i(I_{i,j})$ to Bob.
Bob computes the sampled differences
\[
 \widetilde{\bz}_i(j)=\bx_i(I_{i,j})\oplus\by_i(I_{i,j})
 \qquad(i\in[k],\ j\in[q]).
\]
\item
Bob computes
\[
 h_i=\ind\left\{\frac1q\sum_{j=1}^{q}\widetilde{\bz}_i(j)\ge\frac12\right\}
 \quad \forall i\in[k],
 \qquad
 h=1+\sum_{i=1}^{k}h_i2^{k-i}.
\]
He sends $h_1,\ldots,h_k$ to Alice, so both of them know the address $h\in[m]$.
\item They follow the test in \cref{lem:linear-check} to compute $\cfunc({\bz},\bu_h\oplus \bv_h)$. Let $p_A$ and $p_B$ be the polynomials encoded by
$\bu_h$ and $\bv_h$, respectively, and let $p=p_A+p_B$.
Using public randomness independent of the sampled positions,
they sample $r\in\F_{2^{4k}}$ uniformly.
Alice sends
\[
p_A(r),\qquad f_{(\bx,p_A)}(r),\qquad
g_{(\bx,p_A)}(r),\qquad p_A(\tau_d).
\]
Bob computes
\[
\begin{aligned}
p(r)&=p_A(r)+p_B(r),\\
f_{(\bz,p)}(r)
&=f_{(\bx,p_A)}(r)+f_{(\by,p_B)}(r)
  +f_{(\mathbf{0},p_0)}(r),\\
g_{(\bz,p)}(r)
&=g_{(\bx,p_A)}(r)+g_{(\by,p_B)}(r)
  +g_{(\mathbf{0},p_0)}(r),\\
p(\tau_d)&=p_A(\tau_d)+p_B(\tau_d),
\end{aligned}
\]
and outputs
\[
b
=\ind\big\{p(\tau_d)=0
\ \text{and}\
p(r)=f_{(\bz,p)}(r)g_{(\bz,p)}(r)\big\}.
\]
\end{enumerate}

\textbf{Correctness.} Fix the entire input, and let $\bc_j=\bu_j\oplus\bv_j$ for every $j\in[m]$.
The only randomness is in the sampled positions and $r\sim \F_{2^{4k}}$.
For every $j$ with $\Pr[h=j]>0$, conditioning on $h=j$ fixes the certificate $\bc_j$ and leaves $r$ uniform.
\cref{lem:linear-check} therefore gives
\begin{equation}
\label{eq:conditional-check-error}
 \Pr[b\ne\cfunc(\bz,\bc_j)\mid h=j]<\frac16.
\end{equation}

\textbf{1). Suppose $\indicate(\bz)=1$.}
For each $i\in[k]$, the sampled bits $\widetilde{\bz}_i(1),\ldots,\widetilde{\bz}_i(q)$ are independent and have expectation  $|\bz_i|/m\in [0,1/3]\cup [2/3,1]$.
The corresponding tail bound for independent bits gives
\[
 \Pr[h_i\ne\GH(\bx_i,\by_i)]
 \le \exp\left(-2q\left(\frac16\right)^2\right)
 =e^{-q/18}
 \le\frac1{6k}.
\]
Hence,
\[
 \Pr[h\ne\adrs(\bz)]
 \le\sum_{i=1}^{k}\Pr[h_i\ne\GH(\bx_i,\by_i)]
 \le\frac16.
\]
When $h=\adrs(\bz)$, Definition~\ref{def:total_function} gives
$\tGH=\cfunc(\bz,\bc_h)$.
Equation~\eqref{eq:conditional-check-error} bounds the probability of a wrong answer on this event by less than $1/6$.
Consequently,
\[
\Pr[b\ne\tGH]
\le\Pr[h\ne\adrs(\bz)]
   +\Pr[b\ne\tGH,\ h=\adrs(\bz)]
<\frac16+\frac16=\frac13.
\]

\textbf{2). Suppose $\indicate(\bz)=0$.} By
Definition~\ref{def:total_function}, $\tGH=0$; and Lemma~\ref{lem:check_func} gives
\[
 \cfunc(\bz,\bc_j)=0\qquad\text{for every }j\in[m].
\]
Thus every possible address selects a certificate for which $\cfunc$ is $0$.
Applying \eqref{eq:conditional-check-error} at each such address yields
\[
\begin{aligned}
 \Pr[b\ne\tGH]
 =\sum_{\substack{j\in[m] }}
   \Pr[h=j]\Pr[b=1\mid h=j]
 <\frac16\sum_{\substack{j\in[m]}}\Pr[h=j]
 =\frac16.
\end{aligned}
\]

\textbf{Cost.}
Steps II and III use $kq$ and $k$ bits, respectively.
In Step IV, Alice sends four field elements, each encoded
by $4k$ bits.
The total communication is
$kq+k+16k=O(k\log k)$.
\end{proof}

\section{Rectangle Bounds}
\label{sec:rectangle-bounds}

In this section,
we show that any monochromatic rectangle of $\tGH$ has exponentially small density under the uniform distribution.
Let the input space be
\[
 U=V=(\{0,1\}^{4k(2d-1)})^m,
 \qquad A=\{0,1\}^{km}\times U,
 \qquad B=\{0,1\}^{km}\times V.
\]
Let $\mu_A$ and $\mu_B$ be the uniform measures on $A$ and $B$, and  $\mu=\mu_A\times\mu_B$. A rectangle is a direct product like $A'\times B'$ with $A'\subseteq A$ and $B'\subseteq B$. Its density is $\mu(A'\times B') :=\mu_A(A')\mu_B(B')$. For any $b\in\{0,1\}$, we call  $A'\times B'$ a $b$-rectangle if $\tGH$ equals $b$ on every input in it.
We will prove the following theorem.

\begin{theorem}\label{thm:rectangle_bound}
Under the uniform measure $\mu$, every monochromatic rectangle of $\tGH$ has density at most $2^{-m/(8k)}$.
\end{theorem}
In view of our choice of the parameters, $m=2^k, d = O(km^2)$, thus $m=\poly(n)$ for the input length $n$, establishing~\cref{thm:main}.

\medskip

Before proving Theorem~\ref{thm:rectangle_bound}, we need the following two lemmas. We first show that two large sets of binary strings must contain a pair with small Hamming distance.

\begin{lemma}
\label{lem:distance}
For any integer $d\ge1$, let $E,Q\subseteq\{0,1\}^d$ be nonempty sets, and let
\[
 \alpha :=\frac{|E|}{2^d},\qquad \beta :=\frac{|Q|}{2^d}.
\]
There exist $\ba\in E$ and $\bb\in Q$ such that
\[
 \begin{aligned}
 |\ba\oplus\bb|
 \le\sqrt{\frac d2\ln\frac1\alpha}
      +\sqrt{\frac d2\ln\frac1\beta}
 \le\sqrt{d\ln\frac1{\alpha\beta}}.
 \end{aligned}
\]
\end{lemma}

\begin{proof}
Let $\bZ$ be uniform on $\{0,1\}^d$, and define
\[
 f(\bz)=\min_{\ba\in E}|\bz\oplus\ba|,
 \qquad s=\E[f(\bZ)].
\]
The bits of $\bZ$ are independent, and changing one bit changes $f$ by at most $1$. The bounded differences inequality therefore gives, for every $t\ge0$,
\[
 \begin{aligned}
 \Pr[f(\bZ)-s\le-t]\le e^{-2t^2/d},\quad
 \Pr[f(\bZ)-s\ge t]\le e^{-2t^2/d}.
 \end{aligned}
\]
Since $f(\bz)=0$ exactly when $\bz\in E$, the first inequality with $t=s$ gives
\[
 \alpha=\Pr[f(\bZ)=0]\le e^{-2s^2/d}.
\]
Taking logarithms yields
\[
 s\le\sqrt{\frac d2\ln\frac1\alpha}.
\]
Choose $\bb\in Q$ minimizing $f(\bb)$, and choose $\ba\in E$ such that $|\ba\oplus\bb|=f(\bb)$. If $f(\bb)\le s$, the first claimed bound follows immediately. If $f(\bb)>s$, every $\bz\in Q$ satisfies $f(\bz)\ge f(\bb)$, so the second tail inequality gives
\[
 \beta\le\Pr[f(\bZ)\ge f(\bb)]
 \le e^{-2(f(\bb)-s)^2/d}.
\]
Taking logarithms again yields
\[
 f(\bb)-s\le\sqrt{\frac d2\ln\frac1\beta}.
\]
Adding the bounds on $s$ and $f(\bb)-s$ proves the first claimed inequality. The inequality $(\sqrt a+\sqrt b)^2\le2(a+b)$ for $a,b\ge0$ proves the second one.
\end{proof}

\begin{samepage}
We next show that a large set of tables contains more than half of all possible vectors at most addresses.

\begin{lemma}
\label{lem:projection}
Let $\mathcal U\subseteq U$ be nonempty, and set
$
 \sigma={|\mathcal U|}/{2^{m\cdot4k(2d-1)}}.
$
For each address $h\in[m]$, define the set of vectors that occur at that address by
$
\pi_h(\mathcal U)=\{\bu_h:\bu\in\mathcal U\}.
$
Then at most $\log(1/\sigma)$ addresses contain at most half of all possible vectors:
\[
 \left|\{h\in[m]:|\pi_h(\mathcal U)|\le2^{4k(2d-1)-1}\}\right|
 \le\log( 1/\sigma).
\]
\end{lemma}
\end{samepage}

\begin{proof}
Let
\[
 S=\{h\in[m]:|\pi_h(\mathcal U)|\le2^{4k(2d-1)-1}\}.
\]
Every table in $\mathcal U$ chooses its vector at address $h$ from $\pi_h(\mathcal U)$. Thus
\[
 \begin{aligned}
 |\mathcal U|
\le\prod_{h=1}^m|\pi_h(\mathcal U)|
 \le2^{m\cdot4k(2d-1)-|S|}.
 \end{aligned}
\]
Dividing by $2^{m\cdot4k(2d-1)}$ gives $\sigma\le2^{-|S|}$.
\end{proof}

\subsection{Proof of Theorem~\ref{thm:rectangle_bound}}

We use  Lemma~\ref{lem:distance} and Lemma~\ref{lem:projection} to rule out large $0$-rectangles, and the uniqueness statement in  Lemma~\ref{lem:check_func} to bound the density of $1$-rectangles.

\subsubsection*{$0$-rectangles.}
Suppose a $0$-rectangle $A'\times B'$ has
\[
 \alpha=\mu_A(A'),\qquad \beta=\mu_B(B'),\qquad
 \alpha\beta\ge2^{-m/(8k)}.
\]
We will find a $1$-input in $A'\times B'$, yielding a contradiction. This argument has four steps.

\medskip
\textbf{Step 1.}  We first  keep those $\bx$ and $\by$ that allow many choices of their tables. For each $\bx,\by\in\{0,1\}^{km}$, define the corresponding tables as
\[
 \mathcal U_{\bx}=\{\bu\in U:(\bx,\bu)\in A'\},
 \qquad
 \mathcal V_{\by}=\{\bv\in V:(\by,\bv)\in B'\}.
\]
Define their density by
\[
 \alpha_{\bx}=\frac{|\mathcal U_{\bx}|}{2^{m\cdot4k(2d-1)}},
 \qquad
 \beta_{\by}=\frac{|\mathcal V_{\by}|}{2^{m\cdot4k(2d-1)}}.
\]
Define the strings for which these densities are at least half:
\[
 X=\{\bx:\alpha_{\bx}\ge\alpha/2\},
 \qquad Y=\{\by:\beta_{\by}\ge\beta/2\}.
\]
Considering  the two cases  $X$ and $X^c$, we have
\[
 \alpha=\sum_{\bx} 2^{-km}\cdot \alpha_{\bx}
 \le\frac{|X|}{2^{km}}+\frac\alpha2.
\]
The same argument applies to $Y$, so
\[
 \frac{|X|}{2^{km}}\ge\frac\alpha2,
 \qquad \frac{|Y|}{2^{km}}\ge\frac\beta2.
\]

\medskip
\textbf{Step 2.} \label{step:large-projections}  We find one address where many inputs on each side allow more than half of all possible entries. For each $\bx\in X$ and $\by\in Y$, let
\[
 \begin{aligned}
 S_{\bx}=\{h\in[m]:|\pi_h(\mathcal U_{\bx})|\le2^{4k(2d-1)}/2\},\quad
 T_{\by}=\{h\in[m]:|\pi_h(\mathcal V_{\by})|\le2^{4k(2d-1)}/2\}.
 \end{aligned}
\]
Lemma~\ref{lem:projection} gives
\[
 |S_{\bx}|\le\log\frac1{\alpha_{\bx}}\le\log\frac2\alpha,
 \qquad
 |T_{\by}|\le\log\frac1{\beta_{\by}}\le\log\frac2\beta.
\]
Define the addresses that occur in these sets for more than half of $X$ or $Y$ by
\[
 \begin{aligned}
 S=\{h\in[m]:|\{\bx\in X:h\in S_{\bx}\}|>|X|/2\},\quad
 T=\{h\in[m]:|\{\by\in Y:h\in T_{\by}\}|>|Y|/2\}.
 \end{aligned}
\]
Counting the pairs $(h,\bx)$ with $h\in S_{\bx}$ in two different ways gives
\[
 \begin{aligned}
 \frac{|S|\,|X|}{2}
 \le\sum_{h=1}^m|\{\bx\in X:h\in S_{\bx}\}|
 =\sum_{\bx\in X}|S_{\bx}|
 \le|X|\log\frac2\alpha.
 \end{aligned}
\]
Thus $|S|\le 2\log(2/\alpha)$; counting the pairs $(h,\by)$ in the same way bounds $|T|\le 2\log (2/\beta)$. Consequently,
\[
 \begin{aligned}
 |S\cup T|
 \le2\log\frac2\alpha+2\log\frac2\beta
 =2\log\frac1{\alpha\beta}+4
 \le\frac{m}{4k}+4<m.
 \end{aligned}
\]
So there exists an address  $h\in[m]\setminus(S\cup T)$. Define
\[
 X_h=\{\bx\in X:h\notin S_{\bx}\},
 \qquad Y_h=\{\by\in Y:h\notin T_{\by}\}.
\]
Since $h\notin S\cup T$,
$$
\bigl|\{x\in X:h\in S_x\}\bigr|\le |X|/2,\qquad
\bigl|\{y\in Y:h\in T_y\}\bigr|\le |Y|/2.
$$
Hence
\[
 \frac{|X_h|}{2^{km}}\ge \frac{|X|}{2^{km+1}} \ge  \frac\alpha4,
 \qquad \frac{|Y_h|}{2^{km}}\ge \frac{|Y|}{2^{km+1}}\ge\frac\beta4.
\]
Moreover, their definitions give, for every $\bx\in X_h$ and $\by\in Y_h$,
\[
 |\pi_h(\mathcal U_{\bx})|>2^{4k(2d-1)-1},
 \qquad |\pi_h(\mathcal V_{\by})|>2^{4k(2d-1)-1}.
\]

\medskip
\textbf{Step 3.} We choose $\bx,\by$ with $\indicate(\bx\oplus\by)=1$ and $\adrs(\bx\oplus\by)=h$.
Write $h=1+\sum_{i=1}^k h_i2^{k-i}$, where $h_i\in\{0,1\}$. Let $\mathbf1_m$ be the all-$1$ vector of length $m$. Adding the fixed vector $(h_1\mathbf1_m,\ldots,h_k\mathbf1_m)$ by XOR preserves the size of $Y_h$. Lemma~\ref{lem:distance}, applied to $X_h$ and the set
\[
 \{\by\oplus(h_1\mathbf1_m,\ldots,h_k\mathbf1_m):\by\in Y_h\},
\]
therefore gives $\bx\in X_h$ and $\by\in Y_h$ such that the distance between $\bx$ and $\by\oplus(h_1\mathbf1_m,\ldots,h_k\mathbf1_m)$ is
\begin{equation}\label{eq:address-distance}
 \begin{aligned}
 |\bx\oplus\by\oplus(h_1\mathbf1_m,\ldots,h_k\mathbf1_m)|
 \le\sqrt{km\ln\frac{16}{\alpha\beta}}
 \le m\sqrt{\left(\frac18+\frac{4k}{m}\right)\ln2}\le \frac{m}{3}.
 \end{aligned}
\end{equation}

So each block contributes at most $m/3$. If $h_i=0$,
$
|\bx_i\oplus\by_i|<m/3;
$
if $h_i=1$,
$
|\bx_i\oplus\by_i|>2m/3.
$
The definitions of $\indicate$ and $\adrs$ give
\[
 \indicate(\bx\oplus\by)=1,
 \qquad \adrs(\bx\oplus\by)=h.
\]

\medskip
\textbf{Step 4.}  We choose tables that give output $1$ inside $A'\times B'$, contradicting the assumed output $0$.
Lemma~\ref{lem:check_func} gives a certificate $\bc^*$ such that
\[
 \cfunc(\bx\oplus\by,\bc^*)=1.
\]
By Step 2, more than half of the vectors in $\{0,1\}^{4k(2d-1)}$ occur in $\pi_h(\mathcal U_{\bx})$. The set
\[
 \{\bv_h\oplus\bc^*:\bv_h\in\pi_h(\mathcal V_{\by})\}
\]
also contains more than half of these vectors, since XOR with $\bc^*$ doesn't change the size. If these two sets were disjoint, their union would contain more than $2^{4k(2d-1)}$ vectors in a space of that size. Therefore, they intersect.

By the definition of $\pi_h$, there exist tables $\bu\in\mathcal U_{\bx}$ and $\bv\in\mathcal V_{\by}$ such that
\[
\bu_h\oplus\bv_h=\bc^*.
\]
The definitions of $\mathcal U_{\bx}$ and $\mathcal V_{\by}$ ensure that
\[
 ((\bx,\bu),(\by,\bv))\in A'\times B'.
\]
Since $\indicate(\bx\oplus\by)=1$ and $\adrs(\bx\oplus\by)=h$, Definition~\ref{def:total_function} gives
\[
 \tGH(\bx,\bu;\by,\bv)
 =\cfunc(\bx\oplus\by,\bc^*)=1.
\]
This contradicts the assumption that $A'\times B'$ is a $0$-rectangle.

\subsubsection*{$1$-rectangles.}
We bound the density of all $1$-inputs. Fix $\bx$ and $\by$. If $\indicate(\bx\oplus\by)=0$, then $\tGH=0$ for all tables. If $\indicate(\bx\oplus\by)=1$, the address $h=\adrs(\bx\oplus\by)$ is fixed. Lemma~\ref{lem:check_func} gives exactly one certificate $\bc^*$ for which $\cfunc(\bx\oplus\by,\bc^*)=1$. By Definition~\ref{def:total_function}, the output equals $1$ precisely when
\[
 \bu_h\oplus\bv_h=\bc^*.
\]
Under the uniform distribution on $U\times V$, the vectors $\bu_h$ and $\bv_h$ are independent and uniform on $\{0,1\}^{4k(2d-1)}$. For every choice of $\bu_h$, exactly one of the $2^{4k(2d-1)}$ choices of $\bv_h$ satisfies this equation. Thus, in both cases for $\indicate(\bx\oplus\by)$,
\[
 \Pr_{\bu,\bv}[\tGH(\bx,\bu;\by,\bv)=1]
 \le2^{-4k(2d-1)}.
\]
For any $1$-rectangle $A'\times B'$, every input in it has output $1$. Averaging over $\bx$ and $\by$ therefore gives
\[
 \mu(A'\times B')\le\Pr_{\mu}[\tGH=1]\le2^{-4k(2d-1)}<2^{-m/(8k)}.
\]

\bigskip
\section*{Acknowledgement}
The authors used GPT-6 for exploratory analysis, idea testing,
editorial polishing.
The authors reviewed and finalized all mathematical claims and proofs, and take full responsibility for the correctness, exposition, and attribution in the final manuscript.

\bibliographystyle{alpha}
\bibliography{references}

@inproceedings{yao1979complexity,
  author = {Yao, Andrew Chi-Chih},
  title = {Some Complexity Questions Related to Distributive Computing (Preliminary Report)},
  booktitle = {11th Annual ACM Symposium on Theory of Computing (STOC)},
  pages = {209--213},
  year = {1979},
  publisher = {ACM},
  doi = {10.1145/800135.804414}
}

@inproceedings{boneh2019zero,
  title={Zero-knowledge proofs on secret-shared data via fully linear PCPs},
  author={Boneh, Dan and Boyle, Elette and Corrigan-Gibbs, Henry and Gilboa, Niv and Ishai, Yuval},
  booktitle={Annual international cryptology conference},
  pages={67--97},
  year={2019},
  organization={Springer}
}

@inproceedings{anshu2016separations,
  title={Separations in communication complexity using cheat sheets and information complexity},
  author={Anshu, Anurag and Belovs, Aleksandrs and Ben-David, Shalev and G{\"o}{\"o}s, Mika and Jain, Rahul and Kothari, Robin and Lee, Troy and Santha, Miklos},
  booktitle={2016 IEEE 57th Annual Symposium on Foundations of Computer Science (FOCS)},
  pages={555--564},
  year={2016},
  organization={IEEE}
}

@inproceedings{aaronson2016separations,
  title={Separations in query complexity using cheat sheets},
  author={Aaronson, Scott and Ben-David, Shalev and Kothari, Robin},
  booktitle={Proceedings of the forty-eighth annual ACM symposium on Theory of Computing},
  pages={863--876},
  year={2016}
}

@inproceedings{babai1986complexity,
  author = {Babai, L{\'a}szl{\'o} and Frankl, P{\'e}ter and Simon, J{\'a}nos},
  title = {Complexity Classes in Communication Complexity Theory},
  booktitle = {27th Annual Symposium on Foundations of Computer Science (FOCS)},
  pages = {337--347},
  year = {1986},
  publisher = {IEEE},
  doi = {10.1109/SFCS.1986.15}
}

@inproceedings{impagliazzo2010communication,
  author = {Impagliazzo, Russell and Williams, Ryan},
  title = {Communication Complexity with Synchronized Clocks},
  booktitle = {25th Annual IEEE Conference on Computational Complexity (CCC)},
  pages = {259--269},
  year = {2010},
  publisher = {IEEE},
  url = {https://people.csail.mit.edu/rrw/cc-clocks-conf.pdf}
}

@inproceedings{papakonstantinou2014overlays,
  author = {Papakonstantinou, Periklis and Scheder, Dominik and Song, Hao},
  title = {Overlays and Limited Memory Communication},
  booktitle = {29th Annual IEEE Conference on Computational Complexity (CCC)},
  pages = {298--308},
  year = {2014},
  publisher = {IEEE},
  doi = {10.1109/CCC.2014.37}
}

@article{goos2018landscape,
  author = {G{\"o}{\"o}s, Mika and Pitassi, Toniann and Watson, Thomas},
  title = {The Landscape of Communication Complexity Classes},
  journal = {Computational Complexity},
  volume = {27},
  number = {2},
  pages = {245--304},
  year = {2018},
  note = {Preliminary version in ICALP 2016},
  doi = {10.1007/s00037-018-0166-6}
}

@inproceedings{goos2017query,
  author = {G{\"o}{\"o}s, Mika and Kamath, Pritish and Pitassi, Toniann and Watson, Thomas},
  title = {Query-to-Communication Lifting for {$P^{NP}$}},
  booktitle = {32nd Computational Complexity Conference (CCC 2017)},
  series = {Leibniz International Proceedings in Informatics (LIPIcs)},
  volume = {79},
  pages = {12:1--12:16},
  year = {2017},
  publisher = {Schloss Dagstuhl--Leibniz-Zentrum f{\"u}r Informatik},
  doi = {10.4230/LIPIcs.CCC.2017.12}
}

@inproceedings{chattopadhyay2019equality,
  author = {Chattopadhyay, Arkadev and Lovett, Shachar and Vinyals, Marc},
  title = {Equality Alone Does not Simulate Randomness},
  booktitle = {34th Computational Complexity Conference (CCC 2019)},
  series = {Leibniz International Proceedings in Informatics (LIPIcs)},
  volume = {137},
  pages = {14:1--14:11},
  year = {2019},
  publisher = {Schloss Dagstuhl--Leibniz-Zentrum f{\"u}r Informatik},
  doi = {10.4230/LIPIcs.CCC.2019.14}
}

@article{hatami2026structure,
  author = {Hatami, Hamed and Hatami, Pooya},
  title = {Guest Column: Structure in Communication Complexity and Constant-Cost Complexity Classes},
  journal = {ACM SIGACT News},
  volume = {55},
  number = {1},
  pages = {67--93},
  year = {2024},
  doi = {10.1145/3654780.3654788}
}

@inproceedings{pitassi2020nondeterministic,
  author = {Pitassi, Toniann and Shirley, Morgan and Watson, Thomas},
  title = {Nondeterministic and Randomized Boolean Hierarchies in Communication Complexity},
  booktitle = {47th International Colloquium on Automata, Languages, and Programming (ICALP 2020)},
  series = {Leibniz International Proceedings in Informatics (LIPIcs)},
  volume = {168},
  pages = {92:1--92:19},
  year = {2020},
  publisher = {Schloss Dagstuhl--Leibniz-Zentrum f{\"u}r Informatik},
  doi = {10.4230/LIPIcs.ICALP.2020.92}
}

@misc{gavinsky2026quantum,
  author = {Gavinsky, Dmytro},
  title = {On the Quantum Communication Complexity of Total Functions},
  howpublished = {arXiv:2608.18784v1, \url{https://arxiv.org/abs/2608.18784v1}},
  year = {2026},
  note = {Preliminary manuscript, August 19, 2026}
}

@article{nisan1995rank,
  author = {Nisan, Noam and Wigderson, Avi},
  title = {On Rank vs.\ Communication Complexity},
  journal = {Combinatorica},
  volume = {15},
  number = {4},
  pages = {557--565},
  year = {1995},
  doi = {10.1007/BF01192527},
  url = {https://www.math.ias.edu/~avi/PUBLICATIONS/MYPAPERS/NOAM/RANK/journal.pdf}
}

\appendix
\section{Auxiliary Proofs}
\label{app:auxiliary-proofs}

\begin{proof}[Proof of Lemma~\ref{lem:check_func}]
Run the circuit on $\bz$, and let $a_v^*,b_v^*$ be the actual inputs to its $v$-th $\andgadget$ gate.

\medskip
\textbf{Part 1), existence.}
Suppose $\indicate(\bz)=1$.
Define
\[
p^*(T)=
\left(\sum_{v=1}^d\delta_v(T)a_v^*\right)
\left(\sum_{v=1}^d\delta_v(T)b_v^*\right).
\]
This polynomial has degree at most $2d-2$, so its coefficients define a certificate.
At each $\tau_v$,
\[
p^*(\tau_v)=a_v^*b_v^*.
\]
We prove by induction on $v$ that using these claimed outputs preserves the actual inputs and output of each $\andgadget$ gate.
Before the first such gate, the circuit uses only its original inputs, constants, $\mathsf{NOT}$, and $\mathsf{XOR}$, so both calculations agree.
If the earlier $\andgadget$ outputs agree, applying the same $\mathsf{NOT}$ and $\mathsf{XOR}$ operations gives
\[
a_v(\bz,p^*)=a_v^*,\qquad b_v(\bz,p^*)=b_v^*.
\]
The claimed output $p^*(\tau_v)=a_v^*b_v^*$ is also the actual output, completing the induction.
The definitions of $f$ and $g$ now give
\[
p^*\equiv f_{(\bz,p^*)}g_{(\bz,p^*)}.
\]
The final gate outputs $1-\indicate(\bz)$, so
\[
p^*(\tau_d)=1-\indicate(\bz)=0.
\]
Both requirements hold, and hence $\cfunc(\bz,p^*)=1$.

\medskip
\textbf{Part 1), uniqueness.}
Suppose $\cfunc(\bz,p)=1$.
For every $v\in[d]$, the polynomial identity implies
\[
p(\tau_v)
=f_{(\bz,p)}(\tau_v)g_{(\bz,p)}(\tau_v)
=a_v(\bz,p)b_v(\bz,p).
\]
Again, the inputs to the first $\andgadget$ agree with their actual values.
If the earlier $\andgadget$ outputs agree with their actual values, the intervening $\mathsf{NOT}$ and $\mathsf{XOR}$ operations give $a_v(\bz,p)=a_v^*$ and $b_v(\bz,p)=b_v^*$.
The displayed equality then gives $p(\tau_v)=a_v^*b_v^*$, so the output also agrees.
Induction on $v$ therefore proves these equalities for all $v\in[d]$.
Consequently,
\[
\begin{aligned}
p
\equiv f_{(\bz,p)}g_{(\bz,p)}
\equiv
\left(\sum_{v=1}^d\delta_v a_v^*\right)
\left(\sum_{v=1}^d\delta_v b_v^*\right)
\equiv p^*.
\end{aligned}
\]
Thus all coefficients of $p$ are fixed, and the certificate is unique.

\medskip
\textbf{Part 2).}
Suppose $\indicate(\bz)=0$ and $\cfunc(\bz,p)=1$.
The induction in the uniqueness argument uses only $p\equiv f_{(\bz,p)}g_{(\bz,p)}$, so it still shows that every claimed $\andgadget$ output equals its actual output.
In particular,
\[
p(\tau_d)=1-\indicate(\bz)=1.
\]
This contradicts the  requirement $p(\tau_d)=0$.
Therefore no certificate has $\cfunc(\bz,\bc)=1$.
\end{proof}

\end{document}